\documentclass[review,hidelinks,onefignum,onetabnum]{siamart220329}

\usepackage{amssymb}
\def\E{{\mathbb{E}}}
\def\P{{\mathbb{P}}}

\def \N{{\mathbb{N}}}
\def \I{{\mathbb{I}}}
\def \F{{\mathcal{F}}}
\def \A{{\mathcal{L}}}

\newcommand{\dx}{\operatorname{d}\!}

\usepackage{lipsum}
\usepackage{amsfonts}
\usepackage{graphicx}
\usepackage{epstopdf}
\usepackage{algorithmic}
\ifpdf
  \DeclareGraphicsExtensions{.eps,.pdf,.png,.jpg}
\else
  \DeclareGraphicsExtensions{.eps}
\fi

\newsiamremark{remark}{Remark}
\newsiamremark{hypothesis}{Hypothesis}
\crefname{hypothesis}{Hypothesis}{Hypotheses}
\newsiamthm{claim}{Claim}

\headers{American cancellable option}{Z. Palmowski and P. St\c{e}pniak}

\title{Last passage American cancellable option in L\'evy models
\thanks{Submitted to the editors 01.12.22.
\funding{The research of Zbigniew Palmowski and Pawe{\l} St\c{e}pniak is partially supported by Polish National Science Centre Grant
No. 2021/41/B/HS4/00599.}}}

\author{Zbigniew Palmowski\thanks{Wroc{\l}aw University of Science and Technology, Wroc{\l}aw, Poland
  (\email{zbigniew.palmowski@pwr.edu.pl}
  ).}
\and Pawe{\L} St\c{e}pniak \thanks{Wroc{\l}aw University of Science and Technology, Wroc{\l}aw, Poland
  (\email{Pawel.Stepniak@pwr.edu.pl}).}
}

\usepackage{amsopn}

\begin{document}

\maketitle

% REQUIRED
\begin{abstract}
We derive the explicit price of the perpetual American put option cancelled at the last passage time
of the underlying above some fixed level. We assume the asset process is governed by a geometric
spectrally negative L\'evy process. We show that the optimal exercise time is the first epoch when
asset price process drops below an optimal threshold. We perform numerical analysis as well considering
classical Black-Scholes models and the model where logarithm of the asset price has additional
exponential downward shocks. The proof is based on some martingale arguments and fluctuation theory of L\'evy processes.
\end{abstract}

% REQUIRED

% REQUIRED

\section{Introduction}

The main goal of this paper is to find the closed-form formula for the price of the perpetual American put option
cancelled at the last passage time
of the underlying above some fixed level $h$. More formally, in this paper we find the following value function
\begin{equation}\label{define_V}
\overline{V}(s) = \sup_{\tau \in \mathcal{T}} \E[e^{-r\tau}(K - S_\tau)^+; \tau < \theta |S_0=s],
\end{equation}
for a family of stopping times $\mathcal{T}$, an underlying risky asset price process $S_t$, fixed strike price $K > 0$ and the risk-free interest rate $r$, where
\begin{equation}
\theta = \sup \{ t \geq 0\ |\ S_t \geq h \},
\end{equation}
for some fixed threshold $h > K$. In \eqref{define_V} $\I(\tau < \theta)$ denotes the indicator of the event $\{\tau < \theta\}$.
We assume the L\'evy market, that is, in our model the asset price is described by a geometric spectrally negative L\'evy process
\begin{equation}\label{S_def}
S_t = s e^{X_t},
\end{equation}
where $X_t$ is a spectrally negative L\'evy process and $s = S_0$ is an initial asset price.
It is well known that a L\'evy market allows for a more realistic representation of price dynamics
capturing certain features such as skewness and asymmetry, and a
greater flexibility in calibrating the model to market prices; see e.g. \cite{Cont1, Cont2} and references therein.

In fact, we choose
\begin{equation}\label{X_def}
X_t = x + \mu t + \sigma B_t  - \sum_{k=1}^{N_t}U_k, %\quad \forall_{k}\ U_k \geq 0,
\end{equation}
where $x = X_0 = \log s$ and $\sigma >0$. In \eqref{X_def}
$B_t$ is a Brownian motion, $\mu$ is a~fixed drift,
$N_t$ is an (independent of $B_t$) homogeneous Poisson process with intensity $\lambda$ and $\{U_k\}_{\{k\in \mathbb{N}\}}$ is an (independent of $B_t$ and $N_t$) sequence of independent identically distributed exponential random variables with the expected value $\rho^{-1}$.
We assume that all considered processes live in a common filtered probability space $(\Omega, \F, \{\mathcal{F}_t\}_{\{t\geq 0\}}, \P)$
with a natural filtration $\{\F_t\}_{\{t \geq 0\}}$ of $X_t$ satisfying usual conditions.
Above process \eqref{X_def} makes the asset price process $S_t$ to be a jump-diffusion model.
Additionally, we assume that no dividends are paid to the holders of the underlying asset
and the expectation in \eqref{define_V} is taken with respect of the martingale measure $\P$, that is,
$e^{X_t-rt}$ is a $\P$-local martingale.
In fact, as noted in \cite[Table 1.1, p. 29]{Cont2}, introducing jumps into
the model implies a loss of completeness of the market, which results in the lack of uniqueness
of an equivalent martingale measure. Still, we can choose one of this measure and of course the price
is the same regardless of the choice of the martingale measure.

We allow in \eqref{X_def} to have $\lambda=0$ which corresponds to the classical Black-Scholes (B-S) model
with $\mu =r-\frac{\sigma^2}{2}$.

To find the value function \eqref{define_V}, we will use
the 'guess and verify method' in which we guess the candidate stopping rule and then we verify that this is
truly the optimal stopping rule using the Verification Theorem \ref{verthm}.
That is, we first guess the form of the stopping time
as first downward crossing epoch of some threshold $a>h$
and calculate the value function in terms of the so-called scale functions
using the fluctuation theory of L\'evy processes.
Then we maximize it with respect to the exercise level $a$. In the last step we prove
that our guessed value function satisfies HJB system and hence verification step.

In the final part of the paper we performed extensive numerical analysis based on the known form of the scale functions
for the process \eqref{X_def}.

American option pricing has been investigated over the past four decades in various contexts. This paper
focuses mainly on adding rather new cancellation feature built into the basic American contract.
This cancelling or recalling in the financial contracts can effectively mitigate undesirable positions
in the risky times or the times when the markets are highly volatile.
Therefore we believe that this type of financial contracts can be very attractive for many investors.
Of course this type of derivatives include a cancellation provision.

Our paper continues the research done by \cite{Emmering} and \cite{Kiefer} who consider more complex contracts and
choose continuous time Dynkin game approach. Our approach is more similar to \cite{algo} where authors consider American-style option \eqref{define_V} as well but they do pricing by solving an appropriate HJB system of equation.
In \cite{algo} the underlying asset price was described by the geometric Brownian motion
for which above approach is very natural due to the locality of the diffusive generator of the asset price process $S_t$.
Still, in the context of non-local generators, a 'guess-and-verify' method used in this paper seems to be more efficient.
We also give here a link with seminal HJB equation though.

The paper is organized as follows. Our main result is given in Section
\ref{sec:main}. The proof of the main result is contained in Section \ref{sec:proof}.
At the end of the paper, in Section~\ref{sec:num} we describe some numerical analysis.
%, our new algorithm is in \ref{sec:alg}, numerical
%results are in \ref{sec:experiments}, and the conclusions follow in
%\ref{sec:conclusions}.

\section{Main result}
\label{sec:main}

To present the main result of this paper we introduce required notations.
Let $\P$ be a martingale measure and $\E$ be the expectation with respect of $\P$ with the convention
$\E_{\log s}[\cdot] =\E[\cdot|S_0=s]=\E[\cdot|X_0=\log s]$. We will skip subindex in expectation when $X_0=0$ (hence $S_0=1$).

We define a Laplace exponent of the process $X_t$ via
\begin{equation*}
    \Psi(\theta) = \frac{1}{t}\log\E e^{\theta X_t}.
\end{equation*}
For the process $X_t$ defined in (\ref{X_def}) we have
\begin{equation}\label{Psidef}
    \Psi(\theta) = \mu\theta + \frac{\sigma^2\theta^2}{2} - \frac{\lambda\theta}{\theta + \rho}.
\end{equation}
Since under the risk-neutral measure $\P$, the discounted asset price process $e^{-rt}S_t$ is a~martingale
thus we assume throughout of this paper that
\begin{equation}\label{Psi1}\Psi(1) = r.
\end{equation}
In other words we take
\begin{equation}\label{martmeasure}
  \mu = r - \frac{\sigma^2}{2} +  \frac{\lambda}{1+\rho}.
\end{equation}
%where:
%\begin{equation}
%\lambda = \frac{\E N_t}{t},\ \rho = \frac{1}{\E U_k}.
%\end{equation}

%In order to provide a closed-form formula of the option price, we need to introduce the scale function. Let $\Phi$ be a right-inverse of $\Psi$.
%i.e.:
%\begin{equation}
%		\Phi(x) = \sup\{ \theta \geq 0: \Psi(\theta)=x \},\quad x \geq 0.
%	\end{equation}
For $r \geq 0$ the so-called scale function is defined as a continuous function $W^{(r)}:[ 0,\infty ) \rightarrow [ 0,\infty ) $ such that:
\begin{equation}\label{W_def}
		\int_0^\infty e^{-\beta x}W^{(r)}(x)\dx x = \frac{1}{\Psi(\beta)-r}, \quad \text{for $\beta>\Phi(r)$}.
	\end{equation}
With the first scale function we associate the second one given by
\begin{equation*}
    Z^{(r)}(x) = 1 + r\int_{0}^{x} W^{(r)}(y)\dx y.
\end{equation*}

\begin{lemma}\label{lemma_W}
    The scale function for process $X$ defined by (\ref{X_def}) is given by
    \begin{equation}\label{W_sum}
    W^{(r)}(x) = \sum_{i=1}^3 C_i e^{\eta_i x},
\end{equation}
where:
\begin{align}
    \eta_1 = 1, \qquad
    %\label{eta1} \\
    \eta_{2/3} &= \frac{-1}{2(\rho\sigma^2 + \sigma^2)}\left(2\lambda + 2r + \rho^2\sigma^2 + \rho\sigma^2 + 2r\rho \pm 2\sqrt{\omega}\right) \label{eta2}
%    \eta_3 &= \frac{-1}{2(\rho\sigma^2 + \sigma^2)}\left(2\lambda + 2r + \rho^2\sigma^2 + \rho\sigma^2 + 2r\rho + 2\sqrt{\omega}\right). \label{eta3}
\end{align}
and
\begin{equation}
    \omega = \lambda^2 + \lambda (\rho+1)(2r+\sigma^2) + (\rho+1)^2 \left(r - \frac{1}{2}\rho\sigma^2\right)^2.
\end{equation}
Furthermore,
\begin{align}
    C_1 = \frac{2(\eta_1+\rho)}{\sigma^2(\eta_1 - \eta_2)(\eta_1 - \eta_3)},\quad %\label{C1}\\
    C_2 = \frac{2(\eta_2+\rho)}{\sigma^2(\eta_2 - \eta_1)(\eta_2 - \eta_3)},\label{C2}
    \\
    C_3 = \frac{2(\eta_3+\rho)}{\sigma^2(\eta_3 - \eta_1)(\eta_3 - \eta_2)}. \label{C3}
\end{align}
\end{lemma}
\begin{proof}
From \eqref{Psidef} and \eqref{W_def} it follows that
the scale function is of the form \eqref{W_sum} where $\eta_i$ ($i=1,2,3$) solve $\Psi(\eta)=r$
(see also e.g. \cite{ivanovs, kazu}).
From \eqref{Psi1} we can conclude that $\eta_1$ = 1. Solving remaining square equation we derive (\ref{eta2}). % and (\ref{eta3}).
Observing that (\ref{W_def}) is equivalent to
\begin{equation}
    \frac{\rho + \theta}{\frac{\sigma^2}{2}(\theta - \eta_1)(\theta - \eta_2)(\theta - \eta_3)} =% \frac{1}{\psi(\theta) - r} = \int_0^\infty e^{-\theta x}W^{r}(x)\dx x =
    \sum_{i=1}^3 \frac{C_i}{\theta - \eta_i}
\end{equation}
gives \eqref{C2} and (\ref{C3}).
\end{proof}

We will show that the optimal exercise time is of the form
\begin{equation}\label{tau_def}
    \tau_a = \inf \{ t \geq  0: S_t \leq a \}=\inf \{ t \geq 0: X_t \leq \log(a) \}.
\end{equation}
The threshold $a$ needs to be smaller than the strike price $K$ (and hence of the cancelling threshold $h$)
so that exercising the option can be profitable to the holder.
We will take
\begin{equation}\label{aass}
0<a<K.
\end{equation}
We denote
\begin{equation}\label{functionG}
    G(s) = (K-s)^+\left( \left(\frac{h}{s} \right)^\alpha \wedge 1 \right).
\end{equation}

The main result of this paper is as follows.
\begin{theorem}\label{price}
The price of the perpetual American cancellable put option defined in \eqref{define_V} equals
    	\begin{align*}
    	    \overline{V}(s) &=\frac{\sigma^2}{2}\left[ W^{(r)'}\left( \log\frac{s}{a^*} \right) - \Phi(r)W^{(r)}\left( \log\frac{s}{a^*} \right) \right]G(a^*)  \\
			\\ &+\left[ Z^{(r)}\left( \log\frac{s}{a^*} \right) - \frac{\sigma^2}{2} W^{(r)'}\left( \log\frac{s}{a^*} \right) - W^{(r)}\left( \log\frac{s}{a^*} \right)\left( \frac{r}{\Phi(r)} - \frac{\Phi(r)\sigma^2}{2} \right)  \right]\\
 &\times \int_0^\infty \rho e^{-\rho y}G\left( a^*e^{-y} \right)\dx y
    	\end{align*}
and $\tau_{a^*}$ defined in \eqref{tau_def} is optimal stopping rule for the optimal stopping threshold
%\begin{equation}\label{a*}
%    a^* = K\frac{ \frac{\sigma^2}{2}\left[ C_2\eta_2(\eta_2-1) + C_3\eta_3(\eta_3-1) \right] + \alpha + \frac{\rho}{\rho - \alpha}\sum_{i=2}^3C_i\eta_i\left[ r\left( \frac{1}{\eta_i} - 1 \right) - \frac{\sigma^2}{2}\left( \eta_i - 1 \right)\right] }{ \frac{\sigma^2}{2}\left[ C_2\eta_2(\eta_2-1) + C_3\eta_3(\eta_3-1) \right] - (1-\alpha) + \frac{\rho}{\rho - \alpha + 1}\sum_{i=2}^3C_i\eta_i\left[ r\left( \frac{1}{\eta_i} - 1 \right) - \frac{\sigma^2}{2}\left( \eta_i - 1 \right)\right] }.
%\end{equation}
\begin{equation}\label{a*}
	a^* = \frac{ K\left( \frac{\sigma^2}{2}\sum_{i=2}^3 C_i\eta_i(\eta_i-1) + \alpha + \frac{\rho}{\rho - \alpha}\sum_{i=2}^3C_i\eta_i\left[ r\left( \frac{1}{\eta_i} - 1 \right) - \frac{\sigma^2}{2}\left( \eta_i - 1 \right)\right] \right) }{ \frac{\sigma^2}{2}\sum_{i=2}^3 C_i\eta_i(\eta_i-1) - (1-\alpha) + \frac{\rho}{\rho - \alpha + 1}\sum_{i=2}^3C_i\eta_i\left[ r\left( \frac{1}{\eta_i} - 1 \right) - \frac{\sigma^2}{2}\left( \eta_i - 1 \right)\right] }.
\end{equation}
\end{theorem}

\section{Proof of the main result}\label{sec:proof}
To prove Theorem \ref{price} we start from transforming the value function
$\overline{V}(s)$.
Let $Z_t = \P(\theta > t | \F_t)$ be the conditional survival process. Additionally let us introduce a parameter $\alpha$ solving
\begin{equation}\label{alpha}
\Phi(-\alpha) = 0.
\end{equation}
Above equation has three solutions and the only one that can possibly be negative, that is,
\begin{equation}
\alpha = \frac{\rho}{2} + \frac{\mu}{\sigma^2} - \sqrt{ \left( \frac{\rho}{2} - \frac{\mu}{\sigma^2} \right)^2 + \frac{2\lambda}{\sigma^2}}.
\end{equation}
Observe that $-1 < \alpha < 0$.
Then
%, similarly to \cite{algo}, it can be shown that we can rewrite the definition of process $Z$ in a following form:
\begin{equation*}
  Z_t=\begin{cases}
    \left( \frac{h}{S_t} \right)^\alpha \wedge 1 , & \text{if $\alpha<0$},\\
    1, & \text{if $\alpha \geq 0$}
  \end{cases}
\end{equation*}
and
\begin{equation*}\label{G_1st}
    \E_{\log s}[e^{-r\tau}(K - S_\tau)^+; \tau < \theta ] = \E_{\log s}\left[e^{-r\tau}(K-S_\tau)^+\left( \left(\frac{h}{S_\tau} \right)^\alpha \wedge 1 \right) \right].
\end{equation*}
Hence
\begin{equation}\label{G_2nd}
\overline{V}(s) = \sup_{\tau\in \mathcal{T}} \E_{\log s}[e^{-r\tau}G(S_\tau) ],
\end{equation}
where the function $G$ is defined in \eqref{functionG}.
Above representation is very convenient from the point of
general optimal stopping theory.
We can still modify above representation though.
Observe that by \cite[Thm. 31.5, p. 208]{Sato}
\begin{equation}\label{generator}
    \A f(s) = \tilde\mu s f'(s) + \frac{1}{2} \sigma^2 s^2 f''(s) + \lambda\rho\int_0^\infty \left( f(se^{-y}) - f(s) \right) e^{-\rho y}\dx y
\end{equation}
is the infinitesimal generator of the process $X_t$ acting on $\mathcal{C}^2_0(\mathbb{R})$,
where
\begin{equation}\label{tildemu}
\tilde\mu = r  +  \frac{\lambda}{1+\rho}.
\end{equation}
% and $F_U(x) = 1 - e^{-\rho x}$ is cumulative distribution function of generic $U$.
Due to localization procedure $\A$ is extended generator as well acting on  $\mathcal{C}^2(\mathbb{R})$.
For $s<K$ we denote
\begin{equation}\label{H_def}
    H(s) = \A G(s) - rG(s) = \left(\frac{h}{s} \right)^\alpha (\delta s - rK)\I(s < h) - rK\I(s \geq h),
\end{equation}
where
\begin{equation}
    \delta = \alpha\sigma^2 - \frac{\lambda}{1 + \rho} + \frac{\lambda\rho}{(\rho-\alpha)(1 + \rho - \alpha)}
\end{equation}
and $\I(C)$ denotes the indicator of an event $C$.
Let us also introduce the local time $l_t^{\log (h)}(X)$ of process X at the point $\log (h)$ (see e.g. \cite{peskir3.1}):
\begin{equation}
     l_t^{\log (h)}(X) = \P - \lim_{\varepsilon \downarrow 0 } \frac{1}{2\varepsilon}\int_0^t \I (\log (h)-\varepsilon < X_u < \log (h) + \varepsilon) \dx <X>_u.
\end{equation}
The key representation of $\overline{V}(s)$ is given in next lemma.
\begin{lemma}\label{V=G+V*}%\label{V + G}
 The following holds true:
  \begin{equation}
  \overline{V}(s) = G(s) + V^*(s),
  \end{equation}
  where
%  \begin{equation}\label{V*}
%      V^*(s) = \sup_{\tau\in \mathcal{T}} \E_{\log s}\left[ \int_0^\tau e^{-r\tau}H(S_u) \dx u + \int_0^\tau e^{-r\tau}h\left( G'(h+) - G'(h-) \right) \I (S_u = h) \dx l_t^{\log (h)}(X)\right].
%  \end{equation}
  \begin{align}\label{V*}
\nonumber	V^*(s) &= \sup_{\tau\in \mathcal{T}} \E_{\log s}\left[ \int_0^\tau e^{-r\tau}H(S_u) \dx u \right. \\
&\left. + \int_0^\tau e^{-r\tau}h\left( G'(h+) - G'(h-) \right) \I (S_u = h) \dx l_t^{\log (h)}(X)\right].
\end{align}
\end{lemma}

\begin{proof}
By using the change-of-variable formula \cite[Thm. 3.1]{peskir3.1} we have
%\begin{align}
%     &e^{-r t}G(S_t)=e^{-r t}G(e^{X_t})\nonumber\\ &\quad=  G(e^x)  + \int_0^t e^{-r u} G'(e^{X_u})e^{X_u} \dx X_u^c - \int_0^t r e^{-r u} G(e^{X_u}) \dx u  \nonumber\\
%     \nonumber &\quad + \frac{1}{2}\sigma^2 \int_0^t e^{-r u} \left[ G'(e^{X_u})e^{X_u} + G''(e^{X_u})e^{2X_u} \right] \dx <X_u^c> + \sum_{u \leq t} e^{-r u} \left[ G(e^{X_u}) - G(e^{X_u-}) \right]  \\
%     \nonumber &\quad+ \int_0^t e^{-r u}e^{X_u}\left( G'(h+) - G'(h-) \right) \I (e^{X_u} = h) \dx l_t^{\log (h)}(X)  \\
%     \nonumber & = G(s) + \int_0^t e^{-r u} \left[ \mu S_u G'(S_u) - r G(S_u) + \frac{1}{2} \sigma^2 S_u G'(S_u) + \frac{1}{2} \sigma^2 S_u^2 G''(S_u) \right] \dx u  \\
%     \nonumber &\quad + \int_0^t e^{-r u} S_u G'(S_u) \sigma \dx B_u + \sum_{u \leq t} e^{-r u} \left[ G(S_u) - G(S_u-) \right]  \\
%     \nonumber &\quad +  \int_0^t e^{-r u}h\left( G'(h+) - G'(h-) \right) \I (S_u = h) \dx l_t^{\log (h)}(X)  \\
%     \nonumber &\quad = G(s) + \int_0^t e^{-r u} \left( \A G - rG\right)(S_u) \dx u + \int_0^t e^{-r u}h\left( G'(h+) - G'(h-) \right) \I (S_u = h) \dx l_t^{\log (h)}(X)  \\
%     \nonumber & \quad+ \int_0^t e^{-r u} S_u G'(S_u) \sigma \dx B_u  + \sum_{u \leq t} e^{-r u} \left[ G(S_u) - G(S_u-) \right] \nonumber \\
%      & \quad - \lambda \int_0^t \int_0^\infty \left( G(S_ue^{-y}) - G(S_u) \right)\rho e^{-\rho y}\dx y \dx u. \label{G_ito}
%\end{align}
\begin{align}
	&e^{-r t}G(S_t)=e^{-r t}G(e^{X_t})\nonumber\\ &\quad=  G(e^x)  + \int_0^t e^{-r u} G'(e^{X_u})e^{X_u} \dx X_u^c - \int_0^t r e^{-r u} G(e^{X_u}) \dx u  \nonumber\\
	\nonumber &\quad + \frac{1}{2}\sigma^2 \int_0^t e^{-r u} \left[ G'(e^{X_u})e^{X_u} + G''(e^{X_u})e^{2X_u} \right] \dx <X_u^c> \\
	\nonumber &\quad + \sum_{u \leq t} e^{-r u} \left[ G(e^{X_u}) - G(e^{X_u-}) \right]  \\
	\nonumber &\quad+ \int_0^t e^{-r u}e^{X_u}\left( G'(h+) - G'(h-) \right) \I (e^{X_u} = h) \dx l_t^{\log (h)}(X)  \\
	\nonumber & = G(s) + \int_0^t e^{-r u} \left[ \mu S_u G'(S_u) - r G(S_u) + \frac{1}{2} \sigma^2 S_u G'(S_u) + \frac{1}{2} \sigma^2 S_u^2 G''(S_u) \right] \dx u  \\
	\nonumber &\quad + \int_0^t e^{-r u} S_u G'(S_u) \sigma \dx B_u + \sum_{u \leq t} e^{-r u} \left[ G(S_u) - G(S_u-) \right]  \\
	\nonumber &\quad +  \int_0^t e^{-r u}h\left( G'(h+) - G'(h-) \right) \I (S_u = h) \dx l_t^{\log (h)}(X)  \\
	\nonumber &\quad = G(s) + \int_0^t e^{-r u} \left( \A G - rG\right)(S_u) \dx u \\
	\nonumber &\quad + \int_0^t e^{-r u}h\left( G'(h+) - G'(h-) \right) \I (S_u = h) \dx l_t^{\log (h)}(X)  \\
	\nonumber & \quad+ \int_0^t e^{-r u} S_u G'(S_u) \sigma \dx B_u  + \sum_{u \leq t} e^{-r u} \left[ G(S_u) - G(S_u-) \right] \nonumber \\
	& \quad - \lambda \int_0^t \int_0^\infty \left( G(S_ue^{-y}) - G(S_u) \right)\rho e^{-\rho y}\dx y \dx u. \label{G_ito}
\end{align}

Furthermore, by \cite[eq. (4.34), p. 47]{jacod} and \cite[Thm. 3.4, p. 18 and Rem. 3.5, p. 20]{GS}
the sum of three last increments of \eqref{G_ito}
%$M_t = \int_0^t e^{-r u} S_u G'(S_u) \sigma \dx B_u  + \sum_{u \leq t} e^{-r u} \left[ G(S_u) - G(S_u-) \right]  - \lambda \int_0^t \int_0^\infty \left( G(S_ue^{-y}) - G(S_u) \right)\rho e^{-\rho y}\dx y \dx u $
is a zero-mean martingale. In fact it is a uniformly integrable (UI) martingale.
% since function $G$ given in \eqref{functionG} is bounded.
Indeed, this follows from triangle inequality, fact that $\int_0^t e^{-r u} S_u G'(S_u) \sigma \dx B_u$ is UI martingale
and that
\begin{equation}\label{supUIa}
    \E_{\log s} \left[ \sup_{t \geq 0}  \int_0^{t} e^{-r u}\left| H(S_u) \right| \dx u \right] < \infty
\end{equation}
and
\begin{equation}\label{supUIb}
    \E_{\log s} \left[ \sup_{t \geq 0}  \int_0^{t } e^{-r u}\left|G'(h+) - G'(h-)\right| \I (S_u = h) \dx l_u^{\log (h)}(X) \right] < \infty.
\end{equation}
To show \eqref{supUIa} observe that
from equation (\ref{H_def}) it follows that for $0 < s < h$ function $H(s)$ is continuous and hence bounded and for $s \geq h$ the function $H(s)$ is constant.
To prove \eqref{supUIb} note that
\begin{equation}\label{g'diff}
     G'(h+) - G'(h-) = -\alpha\frac{h-k}{h}.
\end{equation}
Furthermore,
%As $h > K > 0$ and $-1 < \alpha < 0$ we can observe that $0 < -\alpha\frac{h-k}{h} < 1$. Therefore:
\begin{align*}
 \E_{\log s}\left[ \sup_{t \geq 0} \int_0^{t } e^{-r u} \dx l_u^{\log (h)}(X) \right]\leq  \E_{\log s} \left[ \int_0^{\infty } e^{-r u} \dx l_u^{\log (h)}(X) \right]<+\infty.
\end{align*}
Indeed, defining the sequence of consecutive downward passage times of $a$ by
$\tau_1(a)=\tau_a$ and $\tau_{k+1}(a)=\inf\{t >\tau_k: S_t\leq a\}$ and recalling that our price process
$S_t$ is upward skip-free (hence passing upward $\log s$ in a continuous way), from the Markov property of $S_t$ we have
we have
\begin{align*}
&\E_{\log s} \left[ \int_0^{\infty } e^{-r u} \dx l_u^{\log (h)}(X) \right]\leq
\E_{\log s} \left[\int_0^{\tau_a} e^{-r u} \dx l_u^{\log (h)}(X) \right] \left(1+\sum_{k=1}^\infty \E_{\log s} e^{-r\tau_k(a)}\right)\\
&\leq \E_{\log s} \left[\int_0^{\tau_a} e^{-r u} \dx l_u^{\log (h)}(X) \right] \left(1+\sum_{k=1}^\infty \left(\E_{\log s} e^{-r\tau_a}\right)^k\right)<+\infty
\end{align*}
because
\[\E_{\log s} \left[\int_0^{\tau_a} e^{-r u} \dx l_u^{\log (h)}(X) \right]\leq \E_{\log s}  l_{\tau_a}^{\log (h)}(X) <+\infty\]
by \cite[Cor. 3.4]{Bo}.

The proof of the main assertion follows now from Optional Stopping Theorem.
\end{proof}

Next step is the Verification Theorem which allows to identify $V^*(s)$.
\begin{theorem}\label{verthm}
Suppose that function $V\in \mathcal{C}^2(\mathbb{R})$ except the point $h$ and a point $a$
where it is of class $\mathcal{C}^1(\mathbb{R})$. Assume that $V$ satisfies the following HJB system of equations
\begin{align}
    (\A V - rV)(s) &\leq -H(s),
    %\ \text{for}\ s > a,
    \label{system1} \\
%    V(s) \big|_{s = a+} &= 0, \label{V=0} \\
%    V'(s) \big|_{s = a+} &= 0, \label{V'=0} \\
%    V(s) &= 0\ \text{for}\ s < a, \label{system4} \\
%    V(s) &> 0\ \text{for}\ s > a, \label{system5} \\
%    (\A V - rV)(s) &< -H(s)\ \text{for}\ s < a,  \label{system6}\\
    V'(h+) - V'(h-) &= - (G'(h+) - G'(h-)).\label{VGprime}
\end{align}
Then $V(s)\geq V^*(s)$.
\end{theorem}
\begin{proof}
First, we apply the change-of-variable formula to $e^{-r\tau}V(S_t)$ to get%similarly to what we did in (\ref{G_ito}):
%\begin{align}\label{V_ito}
%     e^{-r t}V(e^{X_t}) &=  V(s) + \int_0^t e^{-r u} \left( \A V - rV\right)(S_u) \dx u + \int_0^t e^{-r u}h\left( V'(h+) - V'(h-) \right) \I (S_u = h) \dx l_u^{\log (h)}(X) \\
%     &+\int_0^t e^{-r u}a\left( V'(a+) - V'(a-) \right) \I (S_u = a) \dx l_u^{\log (a)}(X)\nonumber\\
%     \nonumber & + \int_0^t e^{-r u} S_u V'(S_u) \sigma \dx B_u  + \sum_{u \leq t} e^{-r u} \left[ V(S_u) - V(S_u-) \right]\\
%     \nonumber & - \lambda \int_0^t \int_0^\infty \left( V(S_ue^{-y}) - V(S_u) \right)\rho e^{-\rho y}\dx y \dx u.
%\end{align}
\begin{align}\label{V_ito}
	e^{-r t}V(e^{X_t}) &=  V(s) + \int_0^t e^{-r u} \left( \A V - rV\right)(S_u) \dx u \\
	\nonumber & + \int_0^t e^{-r u}h\left( V'(h+) - V'(h-) \right) \I (S_u = h) \dx l_u^{\log (h)}(X) \\
	&+\int_0^t e^{-r u}a\left( V'(a+) - V'(a-) \right) \I (S_u = a) \dx l_u^{\log (a)}(X)\nonumber\\
	\nonumber & + \int_0^t e^{-r u} S_u V'(S_u) \sigma \dx B_u  + \sum_{u \leq t} e^{-r u} \left[ V(S_u) - V(S_u-) \right]\\
	\nonumber & - \lambda \int_0^t \int_0^\infty \left( V(S_ue^{-y}) - V(S_u) \right)\rho e^{-\rho y}\dx y \dx u.
\end{align}

Note that $\int_0^t e^{-r u}a\left( V'(a+) - V'(a-) \right) \I (S_u = a) \dx l_u^{\log (a)}(X)=0$ due to assumed smoothness
of $V$ at $a$.
Further, let $L_t$ is the sum of three last increments of \eqref{V_ito}.
%\begin{equation}
%    N_t = \int_0^t e^{-r u} S_u V'(S_u) \sigma \dx B_u  + \sum_{u \leq t} e^{-r u} \left[ V(S_u) - V(S_u-) \right] - \lambda \int_0^t \int_0^\infty \left( V(S_ue^{-y}) - V(S_u) \right)\dx F_U (y) \dx u.
%\end{equation}
Note that $L_t$ is a mean-one local martingale (see \cite[eq. (4.34), p. 47]{jacod}). Using \eqref{system1}
and \eqref{VGprime}, we can
conclude that
%\begin{align}
%    V(s) + L_t &\geq e^{-r t}V(e^{X_t}) + \int_0^t e^{-r u} H(S_u) \dx u + \int_0^t e^{-r u}h\left( G'(h+) - G'(h-) \right) \I (S_u = h) \dx l_u^{\log (h)}(X) \geq \\
%    \nonumber &\geq \int_0^t e^{-r u} H(S_u) \dx u + \int_0^t e^{-r u}h\left( G'(h+) - G'(h-) \right) \I (S_u = h) \dx l_u^{\log (h)}(X).
%\end{align}
\begin{align}
	V(s) + L_t &\geq e^{-r t}V(e^{X_t}) + \int_0^t e^{-r u} H(S_u) \dx u \\
	\nonumber &+ \int_0^t e^{-r u}h\left( G'(h+) - G'(h-) \right) \I (S_u = h) \dx l_u^{\log (h)}(X)  \\
	\nonumber &\geq \int_0^t e^{-r u} H(S_u) \dx u + \int_0^t e^{-r u}h\left( G'(h+) - G'(h-) \right) \I (S_u = h) \dx l_u^{\log (h)}(X).
\end{align}

Let $(\kappa_n)_{n \in \N}$ be a localizing sequence for $L_t$. Using Optional Stopping Theorem, we can write for any stopping time $\tau$:
%\begin{equation}
%        \E \left[ \int_0^{\tau \wedge \kappa_n} e^{-r u} H(S_u) \dx u + \int_0^{\tau \wedge \kappa_n} e^{-r u}h\left( G'(h+) - G'(h-) \right) \I (S_u = h) \dx l_u^{\log (h)}(X) \right] \leq V(s) + \E [L_{\tau \wedge \kappa_n}] = V(s).
%\end{equation}
\begin{align}
	& \E_s \left[ \int_0^{\tau \wedge \kappa_n} e^{-r u} H(S_u) \dx u + \int_0^{\tau \wedge \kappa_n} e^{-r u}h\left( G'(h+) - G'(h-) \right) \I (S_u = h) \dx l_u^{\log (h)}(X) \right] \\
\nonumber &	\quad  \leq V(s) + \E_s [L_{\tau \wedge \kappa_n}] = V(s).
\end{align}

%If
%\begin{equation}\label{supE}
%    \E \left[ \sup_{t \geq 0} \left| \int_0^{\tau \wedge t} e^{-r u} H(S_u) \dx u + \int_0^{\tau \wedge t } e^{-r u}h\left( G'(h+) - G'(h-) \right) \I (S_u = h) \dx l_u^{\log (h)}(X) \right| \right] < \infty,
%\end{equation}
Now taking the limit with $n$ tending to infinity and applying Lebesgue dominated convergence theorem, we get:
\begin{equation}\label{fatou}
        \E_s \left[ \int_0^{\tau} e^{-r u} H(S_u) \dx u + \int_0^{\tau } e^{-r u}h\left( G'(h+) - G'(h-) \right) \I (S_u = h) \dx l_u^{\log (h)}(X) \right] \leq  V(s)
\end{equation}
which completes the proof.
\end{proof}

Now the main idea is now to find the value function
\[\overline{V}_a(s)=\E_s[e^{-r\tau_a}(K - S_{\tau_a}); \tau_a < \theta ],\qquad s>a,\]
when the exercise time is the first passage downward time $\tau_a$ of the asset price defined in \eqref{tau_def}.
We let
\begin{equation}\label{Vs}
V(s)=\overline{V}_{a^*}(s)-G(s)\quad \text{for $s>a$ and $0$ otherwise}
\end{equation}
for the unique $0<a^*<K$ (see assumption \eqref{aass}) solving
$V_{a^*}(s) \big|_{s = a^*+} = 0$ and $V_{a^*}'(s) \big|_{s = a^*+} = 0$. % (see \eqref{V=0} and \eqref{V'=0}).
In the final step we will show that  $V(s)$ satisfies
all HJB conditions of the Verification Theorem \ref{verthm} and hence we get the assertion of main Theorem
\ref{price}.

We will prove now the following proposition that is interesting in itself.
\begin{proposition}
The value $\overline{V}_a(s)$ equals
\begin{align}
    	    \overline{V}_a(s) &=\frac{\sigma^2}{2}\left[ W^{(r)'}\left( \log\frac{s}{a} \right) - \Phi(r)W^{(r)}\left( \log\frac{s}{a} \right) \right]G(a)  \nonumber\\
			\nonumber\\ &+\left[ Z^{(r)}\left( \log\frac{s}{a} \right) - \frac{\sigma^2}{2} W^{(r)'}\left( \log\frac{s}{a} \right) - W^{(r)}\left( \log\frac{s}{a} \right)\left( \frac{r}{\Phi(r)} - \frac{\Phi(r)\sigma^2}{2} \right)  \right]\nonumber\\
 &\times \int_0^\infty \rho e^{-\rho y}G\left( a e^{-y} \right)\dx y.\label{Vas}
    	\end{align}
\end{proposition}
\begin{proof}
We start the proof from showing that
    \begin{equation}\label{E_D_2a}
        \E \left[ e^{-r\tau_a}; S_{\tau_a}=a \right] = \frac{\sigma^2}{2}\left[ W^{(r)'}\left( \log\frac{s}{a} \right) - \Phi(r)W^{(r)}\left( \log\frac{s}{a} \right) \right]
    \end{equation}
    and
%    \begin{equation}\label{E_D_2}
%        \E \left[ e^{-r\tau_a}; S_{\tau_a}<a \right] = \left[ Z^{(r)}\left( \log\frac{s}{a} \right) - \frac{\sigma^2}{2} W^{(r)'}\left( \log\frac{s}{a} \right) - W^{(r)}\left( \log\frac{s}{a} \right)\left( \frac{r}{\Phi(r)} - \frac{\Phi(r)\sigma^2}{2} \right)  \right].
%    \end{equation}
    \begin{align}\label{E_D_2}
	&\E \left[ e^{-r\tau_a}; S_{\tau_a}<a \right] \\
	\nonumber & \quad  = \left[ Z^{(r)}\left( \log\frac{s}{a} \right) - \frac{\sigma^2}{2} W^{(r)'}\left( \log\frac{s}{a} \right) - W^{(r)}\left( \log\frac{s}{a} \right)\left( \frac{r}{\Phi(r)} - \frac{\Phi(r)\sigma^2}{2} \right)  \right].
\end{align}

Indeed, denoting $\tau_b^-=\inf\{t\geq 0: X_t<b\}$, from  \cite{kazu} we have
	\begin{equation}\label{E_D_Y1}
	    \E_y [e^{-r\tau_0^-}] = Z^{(r)}(y) - \frac{r}{\Phi(r)}W^{(r)}(y)
	\end{equation}
and from  \cite{Kyprianou}:
		\begin{equation}\label{E_D_Y2}
				\mathbb{E}_y [e^{-r\tau_0^-}; X_{\tau_0^-}=0] =  \frac{\sigma^2}{2}\left[ W^{(r)'}\left( y \right) - \Phi(r)W^{(r)}\left( y \right) \right].
		\end{equation}
Now \eqref{E_D_2a} and \eqref{E_D_2} follows directly from the fact that $S_t=se^{X_t}$.

In order to find the option price $\overline{V}(S) =  \E\left[ \text{e}^{-r\tau} G(S_\tau) \right]$ we consider two possible scenarios:
either the underlying price hits the threshold $a$ or it drops below threshold $a$ by a jump. If $S_t$ creeps at $a$,
then $G(S_\tau) = G(a)$ and therefore
% \begin{equation}
%     \E\left[ \text{e}^{-r\tau} G(S_\tau); S_\tau = a \right] = \E\left[ \text{e}^{-r\tau} ; S_\tau = a \right]G(a) = \frac{\sigma^2}{2}\left[ W^{(r)'}\left( \log\frac{s}{a} \right) - \Phi(r)W^{(r)}\left( \log\frac{s}{a} \right) \right]G(a).
% \end{equation}
 \begin{align}
	&\E\left[ \text{e}^{-r\tau} G(S_\tau); S_\tau = a \right] = \E\left[ \text{e}^{-r\tau} ; S_\tau = a \right]G(a) \\
	\nonumber & \quad  = \frac{\sigma^2}{2}\left[ W^{(r)'}\left( \log\frac{s}{a} \right) - \Phi(r)W^{(r)}\left( \log\frac{s}{a} \right) \right]G(a).
\end{align}

In the second scenario $X_{\tau_a} < \log(a)$ and the undershoot $\log(a) - X_{\tau_a}$ has exponential distribution with parameter $\rho$
by the lack of memory of this distribution. Thus
%\begin{align}
%    \nonumber &\E\left[ \text{e}^{-r\tau} G(S_\tau); S_\tau < a \right] = \E\left[ \text{e}^{-r\tau} ; S_\tau < a \right]\E \left[ G(e^{X_\tau}); S_\tau < a \right] =  \E\left[ \text{e}^{-r\tau} ; S_\tau < a \right] \E \left[ G(e^{\log(a) - U }); S_\tau < a \right] =  \\
%    &= \left[ Z^{(r)}\left( \log\frac{s}{a} \right) - \frac{\sigma^2}{2} W^{(r)'}\left( \log\frac{s}{a} \right) - W^{(r)}\left( \log\frac{s}{a} \right)\left( \frac{r}{\Phi(r)} - \frac{\Phi(r)\sigma^2}{2} \right)  \right] \times \int_0^\infty \rho e^{-\rho y}G\left( ae^{-y} \right)\dx y.
%\end{align}
\begin{align}
	 &\E\left[ \text{e}^{-r\tau} G(S_\tau); S_\tau < a \right] = \E\left[ \text{e}^{-r\tau} ; S_\tau < a \right]\E \left[ G(e^{X_\tau}); S_\tau < a \right] \\
	 \nonumber \quad & =  \E\left[ \text{e}^{-r\tau} ; S_\tau < a \right] \E \left[ G(e^{\log(a) - U }); S_\tau < a \right]  \\
	\nonumber \quad & = \left[ Z^{(r)}\left( \log\frac{s}{a} \right) - \frac{\sigma^2}{2} W^{(r)'}\left( \log\frac{s}{a} \right) - W^{(r)}\left( \log\frac{s}{a} \right)\left( \frac{r}{\Phi(r)} - \frac{\Phi(r)\sigma^2}{2} \right)  \right] \\
	\nonumber \quad & \times \int_0^\infty \rho e^{-\rho y}G\left( ae^{-y} \right)\dx y.
\end{align}

Using  $\overline{V}_a(s)=\E_s[e^{-r\tau_a}G(S_{\tau_a})]$ and above identities completes the proof.
\end{proof}
We are now ready to give the proof of the main result of this paper.

\textbf{Proof of Theorem \ref{price}}.
We recall that
\begin{equation}\label{V_def}
    V(s) =
    \begin{cases}
			\overline{V}_{a^*}(s) - G(s), & \text{if}\ s > a \\
			0, & \text{if}\ 0 < s \leq a \\
	\end{cases}
\end{equation}
for $\overline{V}_a(s)$ defined in (\ref{Vas}) and $a^*$ solving $V_{a^*}'(s) \big|_{s = a^*+} = 0$.
We will show that all equation in HJB system given in Verification Theorem \ref{verthm}.

By \cite[Thm.3.10]{scalereview} both scale functions $W^{(r)}$ and $Z^{(r)}$ belongs to $\mathcal{C}^2(\mathbb{R})$.
Hence by \eqref{Vas},  $\overline{V}_{a^*}(s)\in \mathcal{C}^2(\mathbb{R}\setminus\{a^*, h\})$ and of class $\mathcal{C}^1(\mathbb{R})$
at $a^*$ by the choice of $a^*$. Moreover,
\begin{equation}
    V'(h+) - V'(h-) = \overline{V}'(h+) - \overline{V}'(h-) - (G'(h+) - G'(h-)) = - (G'(h+) - G'(h-))
\end{equation}
and hence \eqref{VGprime} is satisfied.

Observe that the only candidate for $0<a<K$ which satisfies condition \newline $V_a'(s) \big|_{s = a+} = 0$ is given
as a solution of the following equation
    \begin{multline*}
        \frac{1}{a} \left(\frac{h}{a}\right)^\alpha \left[ \frac{\sigma^2}{2}(K-a)\left[ C_2\eta_2(\eta_2-1) + C_3\eta_3(\eta_3-1) \right]  (1-\alpha)a + \alpha K  \right. \\
	\left. + \rho \left( \frac{K}{\rho - \alpha} + \frac{a}{\rho - \alpha + 1} \right) \times \sum_{i=2}^3C_i\eta_i\left( r\left( \frac{1}{\eta_i} - 1 \right) - \frac{\sigma^2}{2}\left( \eta_i - 1 \right) \right) \right] = 0
    \end{multline*}
and hence is given in \eqref{a*}.
We still have to verify if $0 < a^* < K$.
To do so we rewrite representation (\ref{a*}) of $a^*$ as follows:
\begin{align}\label{a* mod}
    a^* &
    %= K + K\frac{ 1 + \frac{\rho}{(\rho - \alpha)(\rho - \alpha +1)}\sum_{i=2}^3C_i\eta_i\left[ r\left( \frac{1}{\eta_i} - 1 \right) - \frac{\sigma^2}{2}\left( \eta_i - 1 \right)\right] }{ (\alpha-1) + \sum_{i=2}^3C_i\eta_i \left[ \frac{\rho}{\rho - \alpha + 1}\left( r\left( \frac{1}{\eta_i} - 1 \right) - \frac{\sigma^2}{2}\left( \eta_i - 1 \right)\right) +  \frac{\sigma^2}{2}(\eta_i - 1) \right] } = \\
    %\nonumber &
    = K + K\frac{ 1 + \frac{\rho}{(\rho - \alpha)(\rho - \alpha +1)}\sum_{i=2}^3C_i\eta_i\left[ r\left( \frac{1}{\eta_i} - 1 \right) - \frac{\sigma^2}{2}\left( \eta_i - 1 \right)\right] }{ (\alpha-1) + \sum_{i=2}^3C_i\eta_i \left[ \frac{r\rho}{\rho - \alpha + 1} \left( \frac{1}{\eta_i} - 1 \right) + \frac{\sigma^2}{2}(\eta_i - 1)\left( 1 - \frac{\rho}{\rho - \alpha + 1} \right) \right] }.
\end{align}
Further,
%rom equations (\ref{eta1}, \ref{eta2}) and (\ref{C1}, \ref{C2}) one can learn that:
%\begin{align}
%    C_2\eta_2(\eta_2-1) &= \frac{\rho + 1}{\omega}(\eta_2 + \rho), \qquad
%    %\label{eta2-1} \\
%    C_3\eta_3(\eta_3-1) &= -\frac{\rho + 1}{\omega}(\eta_3 + \rho) \label{eta3-1}
%\end{align}
\begin{equation}
	C_2\eta_2(\eta_2-1) = \frac{\rho + 1}{\omega}(\eta_2 + \rho), \qquad
	%\label{eta2-1} \\
	C_3\eta_3(\eta_3-1) = -\frac{\rho + 1}{\omega}(\eta_3 + \rho) \label{eta3-1}
\end{equation}
and
\begin{equation}\label{eta rho}
    -\eta_2 < \rho < \eta_3;
\end{equation}
see also \cite{kazu}.
Therefore we can see that the right hand sides of equations in %(\ref{eta2-1}) and
(\ref{eta3-1}) are positive, which means that also the left hand sides are also positive. As both $\eta_2$ and $\eta_3$ are negative, it means that both $C_2$ and $C_3$ are also negative.
By virtue of the fact that $\sum_{i=2}^3C_i\left[ r\left( \frac{1}{\eta_i} - 1 \right) - \frac{\sigma^2}{2}\left( \eta_i - 1 \right) \right] = 1$ we can see that:
\begin{equation}
    C_2\left[ r\left( \frac{1}{\eta_2} - 1 \right) - \frac{\sigma^2}{2}\left( \eta_2 - 1 \right) \right] = - C_3\left[ r\left( \frac{1}{\eta_3} - 1 \right) - \frac{\sigma^2}{2}\left( \eta_3 - 1 \right) \right].
\end{equation}
Now by (\ref{eta rho})
\begin{equation}
    C_2\eta_2\left[ r\left( \frac{1}{\eta_2} - 1 \right) - \frac{\sigma^2}{2}\left( \eta_2 - 1 \right) \right] > - C_3\eta_3\left[ r\left( \frac{1}{\eta_3} - 1 \right) - \frac{\sigma^2}{2}\left( \eta_3 - 1 \right) \right].
\end{equation}
which gives
\begin{equation}
    \sum_{i=2}^3C_i\eta_i\left[ r\left( \frac{1}{\eta_i} - 1 \right) - \frac{\sigma^2}{2}\left( \eta_i - 1 \right) \right] > 0.
\end{equation}
This leads to the conclusion that the numerator of the rhs of (\ref{a* mod}) is strictly positive.
Moreover, since $\eta_1$ and $\eta_2$ are negative and $-1<\alpha<0$, we know that its denominator is negative.
This gives immediately that $a^*<K$. To show that
$ a^*>0 $ we need to verify that the numerator plus the denominator is smaller than $0$, that is, that
\begin{align*}
    &1 + \frac{\rho}{(\rho - \alpha)(\rho - \alpha +1)}\sum_{i=2}^3C_i\eta_i\left[ r\left( \frac{1}{\eta_i} - 1 \right) - \frac{\sigma^2}{2}\left( \eta_i - 1 \right)\right]  \\
    \nonumber &+(\alpha-1) + \sum_{i=2}^3C_i\eta_i \left[ \frac{r\rho}{\rho - \alpha + 1} \left( \frac{1}{\eta_i} - 1 \right) + \frac{\sigma^2}{2}(\eta_i - 1)\left( 1 - \frac{\rho}{\rho - \alpha + 1} \right) \right]  \\
    &= \alpha + \sum_{i=2}^3C_i\eta_i \left[ \frac{\rho}{\rho - \alpha}r\left( \frac{1}{\eta_i} - 1 \right) +  \frac{\sigma^2}{2}\left( \eta_i - 1 \right)\left( 1 - \frac{\rho}{\rho - \alpha}\right) \right]<0.
\end{align*}
This follows from the fact that $\alpha, C_2, C_3, \eta_2, \eta_3$ are all strictly negative.

Now, note that $e^{-rt\wedge \tau_{a^*}}V(S_{t\wedge \tau_{a^*}})$ is a martingale. Indeed,
from \cite[Rem. 5]{Avram} we know that $e^{-rt\wedge \tau_{a^*}\wedge \tau_b^+}W^{(r)}(S_{t\wedge \tau_{a^*}\wedge \tau_b^+})$ and
$e^{-rt\wedge \tau_{a^*}\wedge \tau_b^+}Z^{(r)}(S_{t\wedge \tau_{a^*}\wedge \tau_b^+})$ are martingales where $\tau_b^+\inf\{t\geq 0: S_t\geq b\}$.
Furthermore, by \eqref{E_D_2a}, the process \newline $e^{-rt\wedge \tau_{a^*}\wedge \tau_b^+}W^{(r)'}(S_{t\wedge \tau_{a^*}})$ is a martingale as well since
$e^{-rt\wedge \tau_{a^*}\wedge \tau_b^+}A(S_{t\wedge \tau_{a^*}})$ is martingale where $$A(s)=\frac{\sigma^2}{2}\left[ W^{(r)'}\left( \log\frac{s}{a} \right) - \Phi(r)W^{(r)}\left( \log\frac{s}{a} \right) \right]$$ is the right hand of \eqref{E_D_2a}.
To show this observe that by Markov property of $S_t$ we have
%\begin{align*}
%&\E_{\log s}[e^{-r\tau_{a^*}}\I(\tau_{a^*}<\infty, S_{\tau_{a^*}}=a^*)|\mathcal{F}_t]=
%\I(\tau_{a^*}>t)\E_{\log S_t}[e^{-r\tau_{a^*}}\I(\tau_{a^*}<\infty, S_{\tau_{a^*}}=a^*)]\\&+
%I(\tau_{a^*}\leq t)e^{-r\tau_{a^*}}\I(\tau_{a^*}<\infty, S_{\tau_{a^*}}=a^*)
%=e^{-r\tau_{a^*}\wedge t}A(S_{t\wedge \tau_{a^*}}),
%\end{align*}
\begin{align*}
	&\E_{\log s}[e^{-r\tau_{a^*}}\I(\tau_{a^*}<\infty, S_{\tau_{a^*}}=a^*)|\mathcal{F}_t] \\
	&=\I(\tau_{a^*}>t)\E_{\log S_t}[e^{-r\tau_{a^*}}\I(\tau_{a^*}<\infty, S_{\tau_{a^*}}=a^*)]\\&+
	I(\tau_{a^*}\leq t)e^{-r\tau_{a^*}}\I(\tau_{a^*}<\infty, S_{\tau_{a^*}}=a^*)
	=e^{-r\tau_{a^*}\wedge t}A(S_{t\wedge \tau_{a^*}}),
\end{align*}
where we used fact that $A(s)=0$ for $s<a$ and $A(a^*)=\frac{\sigma^2}{2}W^{(r)'}\left( 0 \right)=1$ because $W^{(r)}(0)=0$ due to the assumption that $\sigma >0$ (see \cite[Lem. 3.1 and Lem. 3.2]{scalereview}).
Since $b$ appearing in $\tau_b^+$ above is general, hence
\[(\A \overline{V}_{a^*} - r\overline{V}_{a^*})(s) =0\qquad\text{for $s>a$}
\]
and thus for $s > a$
\begin{equation}
    (\A V - rV)(s) = -(\A G - rG)(s)=-H(s)
\end{equation}
by definition (\ref{H_def}) of $H(s)$.

Now, for $s < a$, $V = 0$ and $(\A V - rV)(s) = 0$.
To prove (\ref{system1}) one need thus to prove that  $H(s) = \A G(s) - rG(s) \leq 0$ for $s < a$.
Using the fact that $a < K < h$ we can write $G(s)$ for $s < a$ in the following form:
\begin{equation}
    G(s) = (K-s)\left(\frac{h}{s} \right)^\alpha.
\end{equation}
Then
\begin{equation}
    \begin{split}
        H(s) = \tilde\mu s G'(s) + \frac{\sigma^2}{2}  s^2 G''(s) + \lambda\int_0^\infty \left( G(se^{-y}) - G(s) \right)\dx F_U (y) - rG(s) \\=
        \left( \frac{h}{s} \right)^\alpha \left[ s\left(\alpha-1\right)\left( r + \frac{\lambda}{1+\rho} - \frac{\sigma^2}{2}\alpha \right) + s\lambda + sr - \frac{s\rho\lambda}{\rho + 1 - \alpha}   \right. \\
        \left. -K (\alpha+1)\left(  r - \frac{\lambda\alpha}{(1+\rho)(\rho-\alpha)} - \frac{\sigma^2}{2}\alpha  \right) \right].
    \end{split}
\end{equation}
It can easily be seen that the term $K(\alpha+1)\left(  r - \frac{\lambda\alpha}{(1+\rho)(\rho-\alpha)} - \frac{\sigma^2}{2}\alpha   \right)$ is strictly positive, as $-1 < \alpha < 0$. Additionally,
$\left(\alpha-1\right)\left( r + \frac{\lambda}{1+\rho} - \frac{\sigma^2}{2}\alpha \right) + \lambda + r - \frac{\rho\lambda}{\rho + 1 - \alpha} $ is strictly negative. As both $h$ and $s$ are positive, hence $H(s) < 0$.
This completes the proof.

\hfill$\square$

\section{Numerical analysis}\label{sec:num}
%\section{Appendix}
%To be written.

\subsection{Geometric Brownian motion}\label{subsec:GBM}

As the first case in the numerical analysis, we consider the underlying asset price described by the geometric Brownian motion. We set the intensity $\lambda$ of the $N_t$ process from equation (\ref{X_def}) equal to zero and thus $X_t$ becomes the arithmetic Brownian motion with drift parameter $\mu = r - \frac{\sigma^2}{2}$. This example corresponds to the option evaluated in \cite{algo}. The scope of the numerical analysis here is to find the optimal exercise level $a^*$ and the fair price $\overline{V}(s)$ of the option. The parameters are chosen as follows: the strike price $K=100$, the threshold $h=120$ so that $h > K$, risk-free rate $r = 5\%$ and the volatility of the underlying asset $\sigma^2 = 0.2$. Additionally the initial price of the underlying asset $s = S_0$ is set to 110. We first start with calculation of $a^*$. Using formula (\ref{a*}) we obtain $a^* = 50$. This value fits well to the assumption (\ref{aass}). Furthermore, when we use formula (43) from \cite{algo}, we get the same result:
\begin{equation}
	a^* = \frac{\eta_2 + \alpha}{\eta_2 + \alpha -1}K = 100\frac{-0.5 - 0.5}{-0.5 - 0.5 - 1} = 50.
\end{equation}	
Now, by using Theorem \ref{price}, we get the price of the cancellable option $\overline{V}(s) \approx 21.76$. Again, when we use the formula (37) proposed in \cite{algo}, we get the same result:
\begin{equation}
	\overline{V}(s) = (K-a^*)\left(\frac{s}{a^*}\right)^{\eta_2}\left(\frac{h}{a^*}\right)^\alpha = \frac{250}{\sqrt{132}} \approx 21.76.
\end{equation}
Now, according to \cite[Chap. 9.2]{Wilmott}, the price of the standard perpetual American option for the no-dividend case is given by:
\begin{equation}
	\overline{V}(s) = Bs^{\alpha^-},
\end{equation}
where $B = -\frac{1}{\alpha^-}\left(\frac{K}{1-1/\alpha^-}\right)^{1-\alpha^-}$ and $\alpha^- = -\frac{2r}{\sigma^2}$. The price of this option with parameters chosen as previously in this section is equal to 36.70. This price is significantly higher than the price of the corresponding cancellable option due to the higher risk the issuer of this contract has to deal with.

Finally,  in Figure \ref{sklejenie_BS} we demonstrate the dependence
of the cancellable option price on the initial price of the underlying instrument.
%The optimal threshold $a$ is constant for all values of $S$. In the picture it is clearly seen how
Note that the price and payoff function fit smoothly at $s = a^*$.

	\begin{figure}\label{sklejenie_BS}
	\includegraphics[width=\textwidth]{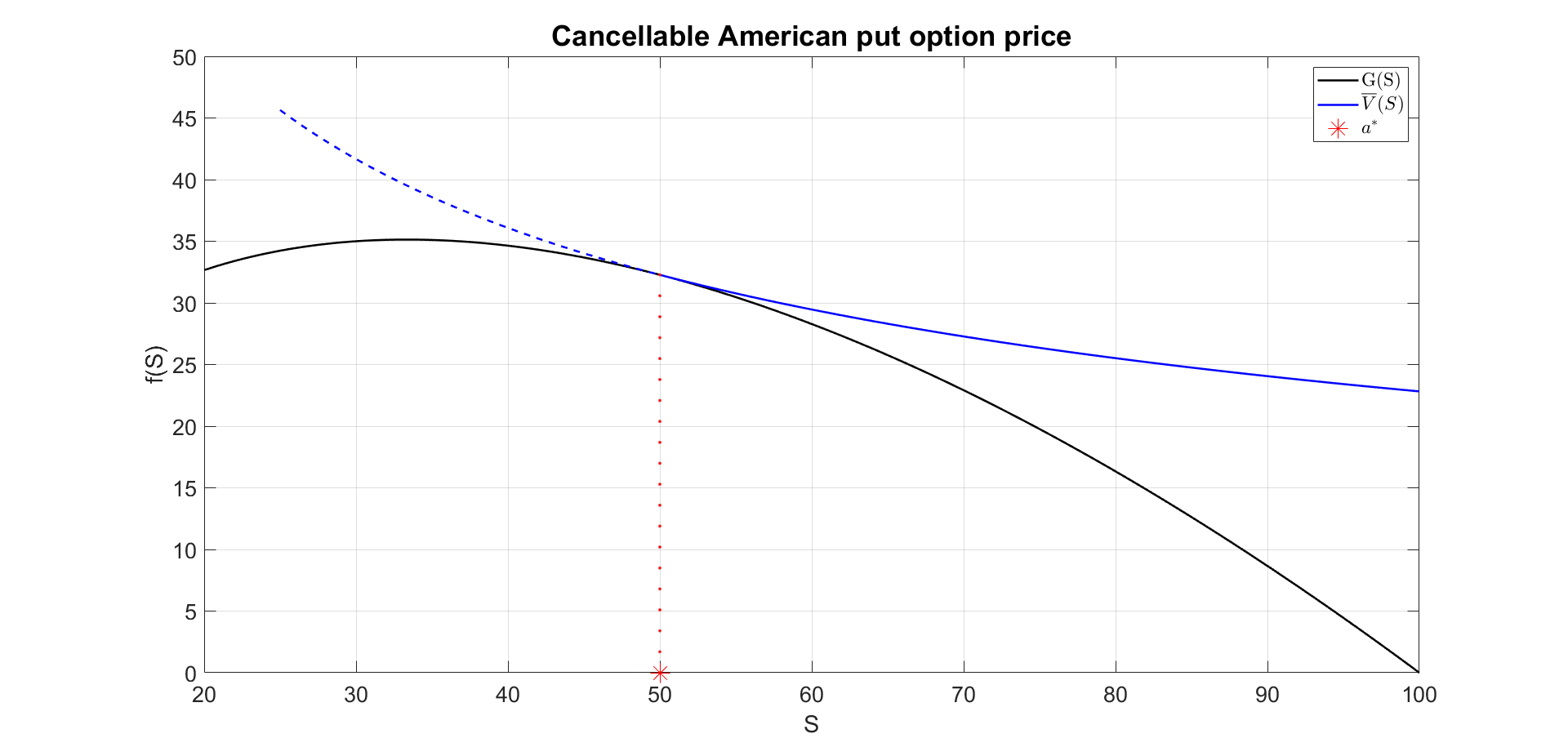}
	\caption{Smooth fit of the payoff and the price functions for the geometric Brownian motion with parameters: $\sigma^2 = 0.2$, $r = 0.05$, $K = 100$, $h = 120$.}
\end{figure}

\subsection{Geometric spectrally negative L\'evy process}

Here we perform similar calculation as in Subsection \ref{subsec:GBM}, but now we set a fixed $\lambda > 0$. We keep the other parameters unchanged, i.e. $h = 120,\ s=S_0 = 110,\ K = 100,\ r = 5\%,\ \sigma^2 = 0.2$ and additionally set $\lambda = 5$ and $\rho = 2$. Again, we start by finding the optimal threshold $a^*$. With formula (\ref{a*}) we obtain $a^* \approx 63.18$. One more time, we use Theorem \ref{price} to find the fair price of the cancellable option and we get $\overline{V}(s) \approx 18.99$.
The price is smaller although a priori it is hard to expect it uniquely.
Indeed, although all the jumps of the underlying are downward the drift is still bigger 
which is a consequence of applying martingale measure in a pricing formula.
%\footnote{PS: nie mam intuicji, czy cena powinna spasc czy wzrosnac w porownaniu do modelu Blacka--Scholesa. Z jednej strony dodalem skoki w dol, ale z drugiej parameter dryfu wzrosl, a zdyskontowana cena instrumentu podstawowego w obu przypadkach pozostaje martyngalem. Ciezko mi sensownie skomentowac ten wynik w zestawieniu z waniliowa opcja wieczysta i cena z modelu BS}

As the final step, in Figure \ref{sklejenie_Levy} we show the behaviour of the payoff and price functions of the cancellable put option, depending on the initial underlying asset price. Again, the smooth fit of the afrementioned functions is clearly visible for $S = a^*$.

	\begin{figure}\label{sklejenie_Levy}
		\includegraphics[width=\textwidth]{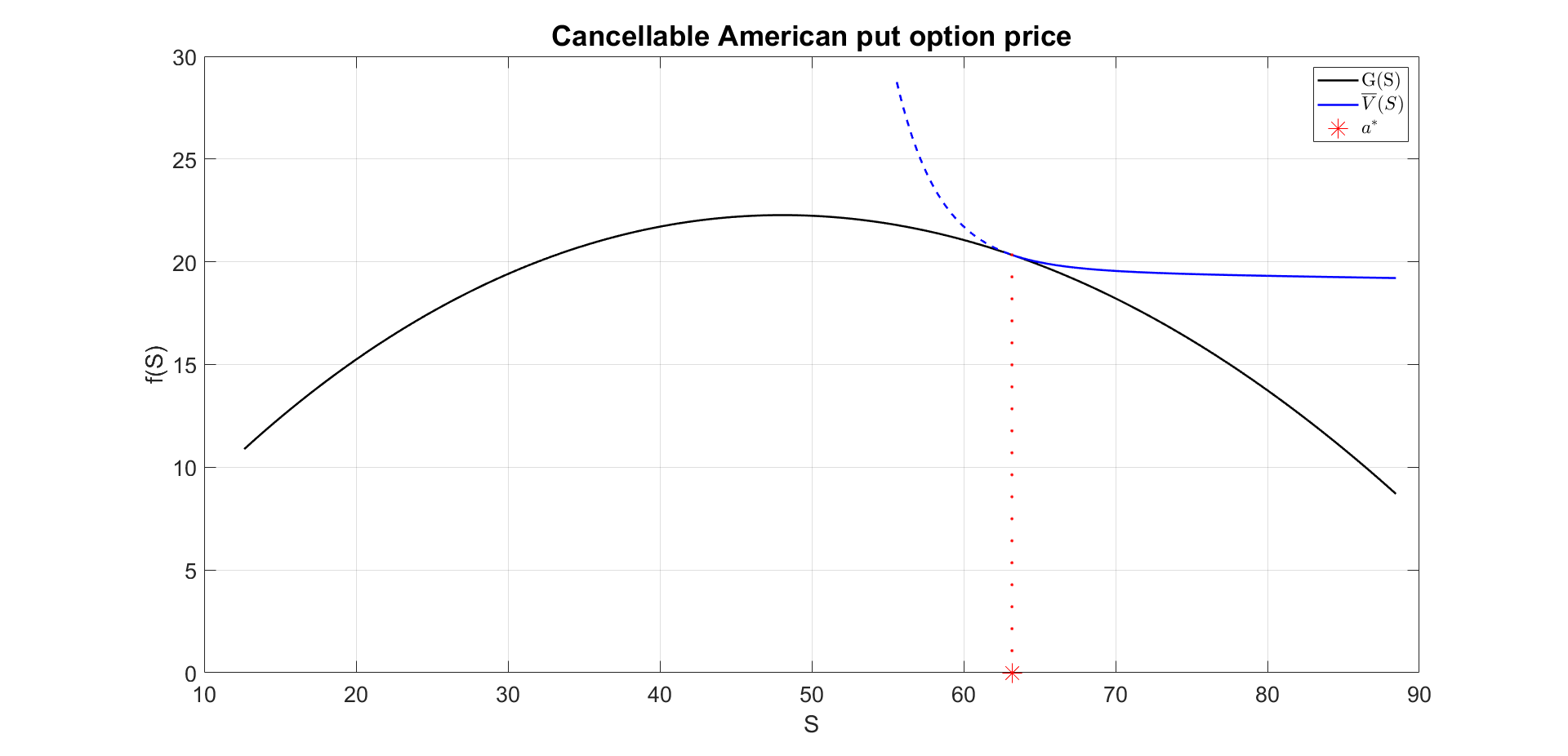}
		\caption{Smooth fit of the payoff and the price functions for: $\rho = 2$, $\sigma^2 = 0.2$, $r = 0.05$, $\lambda = 5$, $K = 100$, $h = 120$.}
	\end{figure}

%\section{Conclusions}
%\label{sec:conclusions}
%
%Some conclusions here.
%
%
%\appendix
%\section{An example appendix}

%\section*{Acknowledgments}
%xxxx

%\bibliography{references}
\end{document}